\documentclass[a4paper,UKenglish]{lipics-v2018}
\hideLIPIcs
\nolinenumbers

\usepackage{microtype}

\usepackage[utf8]{inputenc}
\usepackage{amsmath}
\usepackage{tikz}
\usepackage{caption}
\usepackage{todonotes}

\newcommand{\diam}{\operatorname{diam}}
\newcommand{\rad}{\operatorname{rad}}
\newcommand{\tw}{\operatorname{tw}}
\newcommand{\wien}{\operatorname{wien}}

\renewcommand\path[1]{{\sffamily\small\detokenize{#1}}}

\usepackage{csquotes}

\title{Multivariate Analysis of Orthogonal Range Searching and Graph Distances Parameterized by Treewidth}

\titlerunning{Orthogonal Range Searching and Graph Distances Parameterized by Treewidth}

\author{Karl Bringmann}{Max-Planck-Institute for Informatics, Saarland Informatics Campus, Saarbr\"ucken, Germany}{kbringma@mpi-inf.mpg.de}{}{}

\author{Thore Husfeldt}{BARC, IT University of Copenhagen, Denmark, and Lund University, Sweden.}{thore@itu.dk}{https://orcid.org/0000-0001-9078-4512}{Swedish Research Council grant VR-2016-03855 and Villum Foundation grant 16582.}

\author{Måns Magnusson}{Department of Computer Science, Lund University, Sweden}{mans.magnusson.888@student.lu.se}{}{}

\authorrunning{K. Bringmann and T. Husfeldt and M. Magnusson}

\Copyright{Karl Bringmann and Thore Husfeldt and Måns Magnusson}

\subjclass{Theory of computation $\rightarrow$ Shortest paths, Parameterized complexity and exact algorithms, Computational geometry. Mathematics of computing $\rightarrow$ Paths and connectivity problems.}


\category{}

\relatedversion{}

\supplement{}

\funding{}

\acknowledgements{We thank Amir Abboud and Rasmus Pagh for useful discussions.}

\EventEditors{John Q. Open and Joan R. Access}
\EventNoEds{2}
\EventLongTitle{42nd Conference on Very Important Topics (CVIT 2016)}
\EventShortTitle{CVIT 2016}
\EventAcronym{CVIT}
\EventYear{2016}
\EventDate{December 24--27, 2016}
\EventLocation{Little Whinging, United Kingdom}
\EventLogo{}
\SeriesVolume{42}
\ArticleNo{23}

\begin{document}
\maketitle

\begin{abstract}
  We show that the eccentricities, diameter, radius, and Wiener index of an undirected $n$-vertex graph with nonnegative edge lengths can be computed in time $O(n\cdot \binom{k+\lceil\log n\rceil}{k} \cdot 2^k k^2 \log n)$, where $k$ is the treewidth of the graph.
  For every $\epsilon>0$, this bound is $n^{1+\epsilon}\exp O(k)$, which matches a hardness result of Abboud, Vassilevska Williams, and Wang (SODA 2015) and closes an open problem in the multivariate analysis of polynomial-time computation.
  To this end, we show that the analysis of an algorithm of Cabello and Knauer (Comp. Geom., 2009) in the regime of non-constant treewidth can be improved by revisiting the analysis of orthogonal range searching, improving bounds of the form $\log^d n$ to $\binom{d+\lceil\log n\rceil}{d}$, as originally observed by Monier (J. Alg. 1980). 

  We also investigate the parameterization by vertex cover number.
\end{abstract}

\section{Introduction}

Pairwise distances in an undirected, unweighted graph can be computed by performing a graph exploration, such as breadth-first search, from every vertex.
This straightforward procedure determines the diameter of a given graph with $n$ vertices and $m$ edges in time $O(nm)$.
It is surprisingly difficult to improve upon this idea in general.
In fact, Roditty and Vassilevska Williams \cite{RV} have shown that an algorithm that can distinguish between diameter $2$ and $3$ in an undirected sparse graph in subquadratic time refutes the Orthogonal Vectors conjecture. 

However, for very sparse graphs, the running time becomes linear.
In particular, the diameter of a tree can be computed in linear time $O(n)$ by a folklore result that traverses the graph twice.
In fact, an algorithm by Cabello and Knauer shows that for constant treewidth $k\geq 3$, the diameter (and other distance parameters) can be computed in time $O(n\log^{k-1} n)$, where the Landau symbol absorbs the dependency on $k$ as well as the time required for computing a tree decomposition.
The question raised in \cite{AVW} is how the complexity of this problem grows with the treewidth of the graph.
We show the following result:

\begin{theorem}\label{thm: main}
  The eccentricities, diameter, radius, and Wiener index of a given undirected $n$-vertex graph $G$ of treewidth $\tw (G)$ and nonnegative edge lengths can be computed in time linear in  
  \begin{equation}\label{eq: main bound}
    n \cdot\binom{k + \lceil \log n \rceil}{k} \cdot 2^k k^2 \log n 
  \end{equation}
  where $k=5\tw(G)+4$.
\end{theorem}

For every $\epsilon>0$, the bound \eqref{eq: main bound} is $n^{1+\epsilon}\exp O(\tw (G))$.
This improves the dependency on the treewidth over the running time $n^{1+\epsilon} \exp  O\bigl(\tw(G)\log \tw(G)\bigr)$ of Abboud, Vassilevska Williams, and Wang \cite{AVW}.
Our improvement is tight in the following sense.
Abboud \emph{et al.} \cite{AVW} also showed that under the Strong Exponential Time Hypothesis of Impagliazzo, Paturi, and Zane \cite{IPZ}, there can be no algorithm that computes the diameter with running time 
\begin{equation}\label{eq: AVW}
  n^{2-\delta}\exp{o(\tw(G))}\qquad \text{for any $\delta>0$}\,. 
\end{equation}
In fact, this holds under the potentially weaker Orthogonal Vectors conjecture, see \cite{VassW15} for an introduction to these arguments.
Thus, under this assumption, the dependency on $\tw(G)$ in Theorem~\ref{thm: main} cannot be significantly improved, even if the dependency on $n$ is relaxed from just above linear to just below quadratic.
Our analysis encompasses the Wiener index, an important structural graph parameter left unexplored by~\cite{AVW}.

\medskip

Perhaps surprisingly, the main insight needed to establish Theorem~\ref{thm: main} has nothing to do with graph distances or treewidth.
Instead, we make---or re-discover---the following observation about the running time of $d$-dimensional range trees:

\begin{lemma}[\cite{Monier79}]
  \label{lem: main}
  A $d$-dimensional range tree over $n$ points supporting orthogonal range queries for the aggregate value over a commutative monoid has query time $O(2^d \cdot B(n,d))$ and can be built in time $O(nd\cdot B(n,d))$, where
  \[ B(n,d) = \binom{d+\lceil\log n\rceil}{d}\,.\]
\end{lemma}

This is a more careful statement than the standard textbook analysis, which gives the query time as $O(\log^d n)$ and the construction time as $O(n\log^d n)$.
For many values of $d$, the asymptotic complexities of these bounds agree---in particular, this is true for constant $d$ and for very large $d$, which are the main regimes of interest to computational geometers.
But crucially,  $B(n,d)$ is always $n^{\epsilon} \exp {O(d)}$ for any $\epsilon > 0$, while $\log^d n$ is not.

After Lemma~\ref{lem: main} is realised, Theorem~\ref{thm: main} follows via divide-and-conquer in decomposable graphs, closely following the idea of Cabello and Knauer \cite{CK} and augmented with known arguments \cite{AVW,BDDFLP}.
We choose to give a careful presentation of the entire construction, as some of the analysis is quite fragile.

\medskip

Using known reductions, this implies that the following multivariate lower bound on orthogonal range searching is tight:

\begin{theorem}[Implicit in \cite{AVW}]\label{thm: range query hardness}
  A data structure for the orthogonal range query problem for the monoid $(\mathbf Z,\max)$ with construction time $n\cdot q'(n,d)$ and query time $q'(n,d)$, where
   \[ q'(n,d) = n^{1-\epsilon} \exp o(d)\]
   for some $\epsilon > 0$, refutes the Strong Exponential Time hypothesis.
\end{theorem}

\medskip

We also investigate the same problems parameterized by vertex cover number:

\begin{theorem}\label{thm: vertex cover main}
  The eccentricities, diameter, and radius of a given undirected, unweighted $n$-vertex graph $G$ with vertex cover number $k$ can be computed in time 
    $O(nk +2^kk^2)$.
  The Wiener index can be computed in time $O(nk2^k)$.
\end{theorem}

Both of these bounds are $n\exp O(k)$.
It follows from \cite{AVW} that a lower bound of the form \eqref{eq: AVW} holds for this parameter as well.

\subsection{Related work}

Abboud \emph{et al.}~\cite{AVW} show that given a graph and an optimal tree decomposition, various graph distances can be computed in time \(O( k^2n\log^{k-1} n )\), where $k=\tw (G)$.
This bound is $n^{1+\epsilon} \exp O(k\log k)$ for any $\epsilon >0$. 
This subsumes the running time for finding an approximate tree decomposition with $k=O(\tw (G))$ from the input graph \cite{BDDFLP}, which is $n\exp O(k)$.

If the diameter in the input graph is constant, the diameter can be computed in time $n\exp O(\tw (G))$ \cite{H17}.
This is tight in both parameters in the sense that \cite{AVW} rules out the running time \eqref{eq: AVW} even for distinguishing diameter 2 from 3, and every algorithm needs to inspect $\Omega(n)$ vertices even for treewidth 1.
For non-constant diameter $\Delta$, the bound from \cite{H17} deteriorates as $n\exp O(\tw(G)\log \Delta)$. 
However, the construction cannot be used to compute the Wiener index.

The literature on algorithms for graph distance parameters such as diameter or Wiener index is very rich, and we refer to the introduction of \cite{AVW} for an overview of results directly relating to the present work.
A recent paper by Bentert and Nichterlein~\cite{BN} gives a comprehensive overview of many other parameterisations.

\medskip
Orthogonal range searching using a multidimensional range tree was first described by Bentley~\cite{Bentley80}, Lueker~\cite{Lueker78}, Willard~\cite{Willard}, and Lee and Wong~\cite{LeeW80}, who showed that this data structure supports query time $O(\log^d n)$ and construction time $O(n \log^{d-1}n)$.
Several papers have improved this in various ways by factors logarithmic in $n$; for instance, Chazelle's construction \cite{Chazelle90a} achieves query time $O(\log^{d-1} n)$.

\subsection{Discussion}

In hindsight, the present result is a somewhat undramatic resolution of an open problem in that has been viewed as  potentially fruitful by many people \cite{AVW}, including the second author \cite{H17}.
In particular, the resolution has led neither to an exciting new technique for showing conditional lower bounds of the form $n^{2-\epsilon} \exp{\omega(k)}$, nor a clever new algorithm for graph diameter.
Instead, our solution follows the ideas of Cabello and Knauer \cite{CK} for constant treewidth,  much like in \cite{AVW}.
All that was needed was a better understanding of the asymptotics  of bivariate functions, rediscovering a 40-year old analysis of spatial data structures \cite{Monier79} (see the discussion in Sec.~\ref{sec: range tree discussion}), and using a recent algorithm for approximate tree decompositions \cite{BDDFLP}.

\medskip
Of course, we can derive some satisfaction from the presentation of asymptotically tight bounds for fundamental graph parameters under a well-studied parameterization.
In particular, the surprisingly elegant reductions in \cite{AVW} cannot be improved.
However, as we show in the appendix, when we  parameterize by vertex cover number instead of treewidth,  we can establish even cleaner and tight bounds without much effort.

Instead, the conceptual value of the present work may be in applying the multivariate perspective on high-dimensional computational geometry, reviving an overlooked analysis for non-constant dimension.
To see the difference in perspective, Chazelle's improvement \cite{Chazelle90a} of $d$-dimensional range queries from $\log^d n$ to $\log^{d-1} n$  makes a lot of sense for small $d$, but from the multivariate point of view, both bounds are $n^\epsilon\exp \Omega(d\log d)$.
The range of relationships between $d$ and $n$ where the multivariate perspective on range trees gives some new insight is when $d$ is asymptotically just shy of $\log n$, see Sec.~\ref{sec: asymptotics}.

\medskip

It remains open to find an algorithm for diameter with running time $n\exp O(\tw(G))$, or an argument that such an algorithm is unlikely to exist under standard hypotheses.
This requires better understanding of the regime $d=o(\log n)$.

\section{Preliminaries}

\subsection{Asymptotics}
\label{sec: asymptotics}

We summarise the asymptotic relationships between various functions appearing in the present paper:

\begin{lemma}\label{lem: asymptotics}
  \begin{equation}
      \label{eq: bound 1}
    B(n,d) = O(\log^d n)\,.\\ 
  \end{equation}
  For any $\epsilon > 0$,
    \begin{gather}
    \label{eq: bound 2} 
	B(n,d) = n^\epsilon \exp O(d)\,,\\
      \label{eq: bound 3}
	\log^d n = n^\epsilon \exp \Omega(d\log d)\,, \\
      \label{eq: bound 4} \log^d n = n^\epsilon \exp O(d\log d)\,.
  \end{gather}
\end{lemma}

The first expression shows that $B(n,d)$ is always at least as informative as $O(\log^d n)$.
The next two expressions show that from the perspective of parameterised complexity, the two bounds differ asymptotically:
$B(n,d)$ depends single-exponentially on $d$ (no matter how small $\epsilon>0$ is chosen), while $\log^d n$ does not (no matter how \emph{large} $\epsilon$ is chosen). 
Expression~\eqref{eq: bound 4} just shows that \eqref{eq: bound 3} is maximally pessimistic.

\begin{proof}
  Write $h=\lceil\log n\rceil$.
To see \eqref{eq: bound 1}, consider first the case where $d<h$.
  Using $\binom{a}{b}\leq a^b/b!$ we see that
\begin{equation}\label{eq: derivation small d}
  \binom{d+h}{d}\leq \binom{2h}{d}  \leq \frac{(2h)^d}{d!} = \frac{2^d}{d!}h^d=O( \log^d n)\,.
\end{equation}
Next, if $d\geq h$ then
\[ 
  \binom{d+h}{d}=\binom{d+h}{h} \leq \binom{2d}{h} = \frac{2^h}{h!} d^h \leq d^h\,,\]
provided $h\geq 4$.
It remains to observe that $d^h\leq h^d=O(\log^d n)$.
Ineed, since the function $\alpha\mapsto \alpha/\ln\alpha$ is increasing for $\alpha \geq \mathrm e$, we have $h/\ln h \leq d\ln d$, which implies $\exp(h\ln d)\leq \exp(d\ln h)$ as needed.

For \eqref{eq: bound 2}, there are two cases.
First assume $d < \epsilon h$ for all $\epsilon>0$.
From Stirling's formula we know $\binom{a}{b} \leq \big( \frac{\mathrm{e}a}{b} \big)^k$, so
\[ \binom{d + h}{d} <
\binom{(1+\epsilon) h}{\epsilon h} <
\Big(\frac{\mathrm{e} (1+\epsilon) h}{\epsilon h}\Big)^{\epsilon h} <
\Big(\frac{\mathrm{e}(1+\epsilon)}{\epsilon}\Big)^{2\epsilon \log n} =
n^{2\epsilon \log \mathrm{e}(1 + \epsilon)\epsilon^{-1}}  =
n^{o(1)}\,, \]
  where the last expression uses that $\epsilon\mapsto 2\epsilon \log \mathrm{e}(1 + \epsilon)\epsilon^{-1}$ is a monotone increasing function in the interval $\big(0, \frac{1}{2} \big]$.

On the other hand, if $d \geq c h$ for some constant $c$, we have
  \[\binom{d+h}{d} \leq \binom{(1+1/c)d}{d} < 2^{(1+1/c)d} = \exp O(d)\,.\]

We turn to \eqref{eq: bound 3}.
Assume that there is a function $g$ such that
\[ \log^d n = n^c g(d)\,.\]
Then choose $b > 1$ and consider $d$ such that
  \( d =  b^{-1}\log n\,. \)
Then
\[ g(d) \geq \frac{\log^d n}{n^c} = 2^{d\log \log n -c\log n}= 2^{d\log{(bd)} - cbd} = \exp\Omega(d\log d) \,.
\]

  Finally for \eqref{eq: bound 4}, we repeat the argument from \cite{AVW}.
  If $d\leq \epsilon \log n/\log\log n$ then
\( \log^d n = 2^{d\log \log n} \leq n^\epsilon  \,.\)
In particular, if $d=o(\log n/\log\log n)$ then $\log ^d n = n^{o(1)}$.
Moreover, for $d\geq \log^{1/2} n$ we have $\log \log n \le 2 \log d$ and thus 
\( \log^d n = 2^{d \log \log n} \le 4^{d \log d}. \)
\end{proof}

These calculations also show the regimes in which these considerations are at all interesting.
For $d=o(\log n/\log\log n)$ then both functions are bounded by $n^{o(1)}$, and the multivariate perspective gives no insight.
For $d\geq \log n$, both bounds exceed $n$, and we are better off running $n$ BFSs for computing diameters, or passing through the entire point set for range searching.

\subsection{Model of computation}

We operate in the word RAM, assuming constant-time arithmetic operations on coordinates and edge lengths, as well as constant-time operations in the monoid supported by our range queries.
For ease of presentation, edge lengths are assumed to be nonnegative integers; we could work with abstract nonnegative weights instead \cite{CK}.

\section{Orthogonal Range Queries}
\label{sec: Orthogonal Range Queries}

\subsection{Preliminaries}

Let $P$ be a set of $d$-dimensional points.
We will view $p\in P$ as a vector $p=(p_1,\ldots, p_d)$.

A commutative \emph{monoid} is a set $M$ with an associative and commutative binary operator $\oplus$ with identity.
The reader is invited to think of $M$ as the integers with $-\infty$ as identity and $a\oplus b = \max\{a,b\}$.

Let $f\colon P\rightarrow M$ be a function and define for each subset $Q\subseteq P$
\[ f(Q) = \bigoplus\{\, f(q)\colon q\in Q\}\,\]
with the understanding that $f(\emptyset)$ is the identity in $M$.

\begin{figure}[ht]
\begin{minipage}{4cm}
\begin{tikzpicture}[scale = .75,>=stealth]
  \pgfsetxvec{\pgfpoint{-.67cm}{-.5cm}}
  \pgfsetyvec{\pgfpoint{1cm}{-.1cm}}
  \pgfsetzvec{\pgfpoint{0cm}{1cm}}
  \draw [thick,->] (0,0,0) -- (2.5,0,0) node [left] {$x$};
  \draw [thick,->] (0,0,0) -- (0,2.5,0) node [right] {$y$};
  \draw [thick,->] (0,0,0) -- (0,0,2.5) node [above] {$z$};
  \draw (0,0,1) -- (0,-.2,1);
  \draw (0,0,2) -- (0,-.2,2);
  \draw (0,1,0) -- (-.2,1,0);
  \draw (0,2,0) -- (-.2,2,0);
  \draw (1,0,0) -- (1,-.2,0);
  \draw (2,0,0) -- (2,-.2,0);
  \node (p) at (0,0,0) [circle, inner sep = 1pt, fill, label = 45:$p$] {};
  \node (q) at (2,0,0) [circle, inner sep = 1pt, fill, label = above:$q$] {};
  \node (r) at (0,2,1) [circle, inner sep = 1pt, fill, label = right:$r$] {};
  \draw [gray!50] (0,2,0) -- (r);
  \node (s) at (2,1,2) [circle, inner sep = 1pt, fill, label = left:$s$] {};
  \draw [gray!50] (2,1,0) -- (s);
  \draw [gray!50] (0,1,0) -- (2,1,0);
  \draw [gray!50] (2,0,0) -- (2,1,0);
\end{tikzpicture}
\end{minipage}
\begin{minipage}{4cm}
\begin{tabular}{lll}
  $p$ & $(0,0,0)$ & $f(p) = 5$ \\
  $q$ & $(2,0,0)$ & $f(q) = 6$ \\
  $r$ & $(0,2,1)$ & $f(r) = 7$ \\
  $s$ & $(2,1,2)$ & $f(s) = 8$ \\
\end{tabular}
\end{minipage}

  \caption{\label{fig: points}
  Four points in three dimensions.
  With the monoid $(\mathbf Z,\max)$ we have $f(\{p,r,s\}) = 7$.
  }
\end{figure}
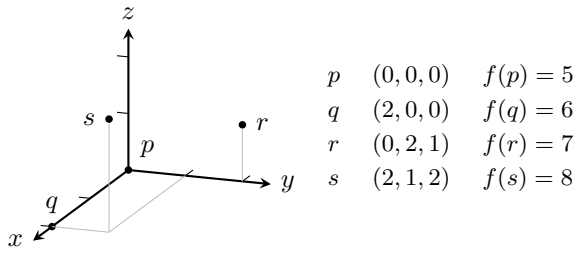

\subsection{Range Trees}
\label{sec: Range Trees}
\newcommand{\med}{\operatorname{med}}

Consider dimension $i\in\{1,\ldots,d\}$ and enumerate the points in $Q$ as $q^{(1)},\ldots,q^{(r)}$ such that $q^{(j)}_i \leq q^{(j+1)}_i$, for instance by ordering after the $i$th coordinate and breaking ties lexicographically.
Define $\med_i (Q)$ to be the \emph{median} point $q^{(\lceil r/2\rceil)}$, and similarly the $\min_i(Q) = q^{(1)}$ and $\max_i(Q)= q^{(r)}$.
Set 
\begin{equation}\label{eq: Q_L and Q_R def} 
  Q_L = \{q^{(1)},\ldots, q^{(\lceil r/2\rceil)}\},\qquad
  Q_R = \{q^{(1+\lceil r/2\rceil)},\ldots, q^{(r)}\}\,.
\end{equation}

For $i\in\{1,\ldots, d\}$, the \emph{range tree} $R_i(Q)$ for $Q$ is a node $x$ with the following attributes:

  \begin{itemize}
    \item  $L[x]$, a reference to the range tree $T_i(Q_L)$, often called the \emph{left child} of $x$.
    \item  $R[x]$, a reference to the range tree $T_i(Q_R)$, often called the \emph{right child} of $x$.
    \item  $D[x]$, a reference to the range tree $T_{i+1}(Q)$, often called the \emph{secondary}, \emph{associate}, or \emph{higher-dimensional} structure. 
    This attribute only exists for $i<d$.
    \item  $l[x]=\min_i(Q)$. 
    \item  $r[x]=\max_i(Q)$.
    \item  $f[x]= f(Q)$.
    This  attribute only exists for $i=d$.
\end{itemize}

\paragraph*{Construction}

Constructing a range tree for $T$ is a straightforward recursive procedure:

\bigskip\noindent$\blacktriangleright$
{\sffamily\bfseries Algorithm C}
({\sffamily Construction}).
{\it 
Given integer $i\in\{1,\ldots, d\}$ and a list $Q$ of points, this algorithm constructs the range tree $R_i(Q)$ with root $x$.
}

\vspace*{-2ex}
\begin{description}
  \item[C1] [Base case $Q = \{q\}$.] Recursively construct $D[x]=T_{i+1}(Q)$ if $i<d$, otherwise set $f[x]=f(q)$. Set $l[x] = r[x] = q_i$. Return $x$.
  \item[C2] [Find median.] Determine $q=\med_i Q$, $l[x]=\min_i(Q)$, $r[x]=\max_i(Q)$.
  \item[C3] [Split $Q$.] Let $Q_L$ and $Q_R$ as given by \eqref{eq: Q_L and Q_R def}, note that both are nonempty.
  \item[C4] [Recurse.]
    Recursively construct $L[x]=R_i(Q_L)$ from $Q_L$.
    Recursively construct $R[x]=R_i(Q_R)$ from $Q_R$.
    If $i<d$ then recursively construct $D[x]=T_{i+1}(Q)$.
    If $i=d$ then set $f[x]=  f[L[x]]\oplus f[R[x]]$.
\end{description}

The data structure can be viewed as a collection of binary trees whose nodes $x$ represent various subsets $P_x$ of the original point set $P$.
In the interest of analysis, we now introduce a scheme for naming the individual nodes $x$, and thereby also the subsets $P_x$.
Each node $x$ is identified by  a string of letters from $\{\mathrm L, \mathrm R,\mathrm D\}$ as follows.
Associate with $x$ a set of points, often called the \emph{canonical subset} of $x$, as follows.
For the empty string $\epsilon$ we set $P_\epsilon= P$.
In general, if $Q=P_x$ then $P_{x\mathrm L} = Q_L$, $P_{x\mathrm R} = Q_R$ and $P_{x\mathrm D} = Q$.
The strings over $\{\mathrm L, \mathrm R,\mathrm D\}$ can be understood as uniquely describing a path through in the data structure; 
for instance, L means `go \emph{left}, i.e., to the left subtree, the one stored at $L[x]$' and D means `go to the next \emph{dimension}, i.e., to the subtree stored at $D[x]$.'
The name of a node now describes the unique path that reaches it.

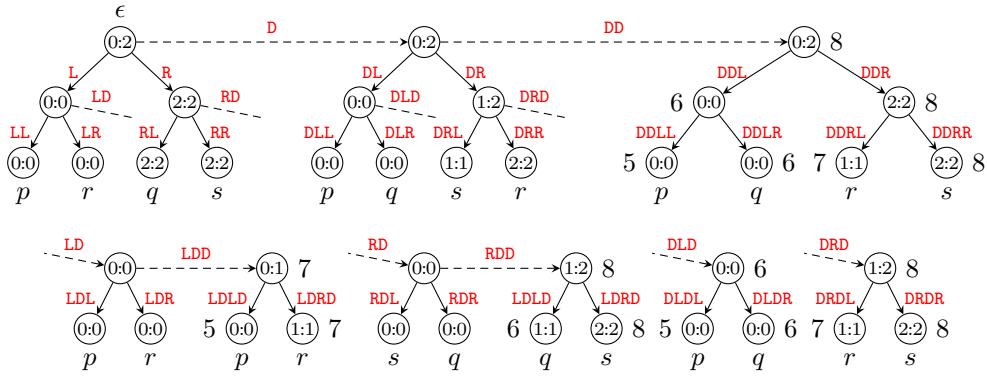
\begin{figure}
\usepgflibrary{shapes.multipart}
\tikzstyle{vertex}=[draw, circle, font=\scriptsize, inner sep = 0.5pt]
\tikzstyle{edge from parent}=[->,draw]
\tikzstyle{xxx}=[red, font = \scriptsize\tt]
\begin{tikzpicture}[>=stealth,level distance=8mm,level/.style={sibling distance=17mm/#1},
  every edge from parent node/.style={text=red, font=\scriptsize\tt}]

  \node (eps) at (0,0) [vertex, label = $\epsilon$] {$0\!\!:\!\!2$}
    child {node (L)  [vertex] {$0\!\!:\!\!0$}
      child {node (LL) [vertex, label = below:$p$] {$0\!\!:\!\!0$} edge from parent node [xxx, left] {LL}}
      child {node (LR) [vertex, label = below:$r$] {$0\!\!:\!\!0$} edge from parent node [xxx, right] {LR}}
     edge from parent node [xxx,left] {L}
    }
    child {node (R) [vertex] {$2\!\!:\!\!2$}
      child {node (RL) [vertex, label = below:$q$] {$2\!\!:\!\!2$} edge from parent node [xxx, left ] {RL}}
      child {node (RR) [vertex, label = below:$s$] {$2\!\!:\!\!2$} edge from parent node [xxx, right] {RR}}
      edge from parent node [xxx, right] {R}
    };
  \node (D) at (4,0) [vertex] {$0\!\!:\!\!2$}
    child {node (DL)  [vertex] {$0\!\!:\!\!0$}
      child {node (DLL) [vertex, label = below:$p$] {$0\!\!:\!\!0$} edge from parent node [xxx, left] {DLL}}
      child {node (DLR) [vertex, label = below:$q$] {$0\!\!:\!\!0$} edge from parent node [xxx, right] {DLR}}
     edge from parent node [xxx,left] {DL}
    }
    child {node (DR) [vertex] {$1\!\!:\!\!2$}
      child {node (DRL) [vertex, label = below:$s$] {$1\!\!:\!\!1$} edge from parent node [xxx, left ] {DRL}}
      child {node (DRR) [vertex, label = below:$r$] {$2\!\!:\!\!2$} edge from parent node [xxx, right] {DRR}}
      edge from parent node [xxx, right] {DR}
    };

    \begin{scope}[level/.style={sibling distance=25mm/#1}]
  \node (DD) at (9,0) [vertex, label= right:8] {$0\!\!:\!\!2$}
    child {node (DDL)  [vertex, label= left: 6] {$0\!\!:\!\!0$}
      child { node (DDLL) [vertex, label = left: 5, label = below:$p$] {$0\!\!:\!\!0$} 
	      edge from parent node [xxx, left] {DDLL}}
      child { node (DDLR) [vertex, label = right: 6, label = below:$q$] {$0\!\!:\!\!0$} 
	      edge from parent node [xxx, right] {DDLR}}
     edge from parent node [xxx,left] {DDL}
    }
    child {node (DDR) [vertex, label = right: 8] {$2\!\!:\!\!2$}
      child { node (DDRL) [vertex, label = left: 7, label = below:$r$] {$1\!\!:\!\!1$} 
	      edge from parent node [xxx, left ] {DDRL}}
      child { node (DDRR) [vertex, label = right: 8, label = below:$s$] {$2\!\!:\!\!2$} 
	      edge from parent node [xxx, right] {DDRR}}
      edge from parent node [xxx, right] {DDR}
    };
    \end{scope}

    \begin{scope}[level/.style={sibling distance=8mm/#1}]
      \node  (LD) at (0,-3) [vertex] {$0\!\!:\!\!0$}
	child {
	  node (LDL)  [vertex, label = below:$p$] {$0\!\!:\!\!0$}
	  edge from parent node [xxx,left] {LDL}
	  }
	child {node (LDR) [vertex, label = below:$r$] {$0\!\!:\!\!0$}
	  edge from parent node [xxx, right] {LDR}
	};
    \draw [densely dashed] (L) -- node [xxx, above] {LD} +(1,-.2);
    \draw [<-, densely dashed] (LD) -- node [xxx, above] {LD} +(-1,.2);

      \node (LDD) at (2,-3) [vertex, label = right: 7] {$0\!\!:\!\!1$}
	child {node (LDLD)  [vertex, label = left: 5, label = below:$p$] {$0\!\!:\!\!0$}
	  edge from parent node [xxx,left] {LDLD}
	  }
	child {node (LDRD) [vertex, label = right: 7, label = below:$r$] {$1\!\!:\!\!1$}
	  edge from parent node [xxx, right] {LDRD}
	};
      \draw [->, densely dashed] (LD) -- node [xxx, above] {LDD} (LDD);

      \node (RD) at (4,-3) [vertex] {$0\!\!:\!\!0$}
	child {node (RDL)  [vertex, label = below:$s$] {$0\!\!:\!\!0$}
	  edge from parent node [xxx,left] {RDL}
	  }
	child {node (RDR) [vertex, label = below:$q$] {$0\!\!:\!\!0$}
	  edge from parent node [xxx, right] {RDR}
	};
      \draw [densely dashed] (R) -- node [xxx, above] {RD} +(1,-.2);
      \draw [<-, densely dashed] (RD) -- node [xxx, above] {RD} +(-1,.2);
      \node (RDD) at (6,-3) [vertex, label = right: 8] {$1\!\!:\!\!2$}
	child {node (RDLD)  [vertex, label = left: 6, label = below:$q$] {$1\!\!:\!\!1$}
	  edge from parent node [xxx,left] {LDLD}
	  }
	child {node (RDRD) [vertex, label = right: 8, label = below:$s$] {$2\!\!:\!\!2$}
	  edge from parent node [xxx, right] {LDRD}
	};
    \draw [->, densely dashed] (RD) -- node [xxx, above] {RDD} (RDD);
    \end{scope}

    \begin{scope}[level/.style={sibling distance=8mm}]
      \node (DLD) at (8,-3) [vertex, label = right: 6] {$0\!\!:\!\!0$}
	child {node (DLDL)  [vertex, label = left: 5, label = below:$p$] {$0\!\!:\!\!0$}
	  edge from parent node [xxx,left] {DLDL}
	  }
	child {node (DLDR) [vertex, label = right: 6, label = below:$q$] {$0\!\!:\!\!0$}
	  edge from parent node [xxx, right] {DLDR}
	};
	\draw [densely dashed] (DL) -- node [xxx, above] {DLD} +(1,-.2);
	\draw [<-, densely dashed] (DLD) -- node [xxx, above] {DLD} +(-1,.2);
    \end{scope}

    \begin{scope}[level/.style={sibling distance=8mm}]
      \node (DRD) at (10,-3) [vertex, label = right: 8] {$1\!\!:\!\!2$}
	child {node (DRDL)  [vertex, label = left: 7, label = below:$r$] {$1\!\!:\!\!1$}
	  edge from parent node [xxx,left] {DRDL}
	  }
	child {node (DRDR) [vertex, label = right: 8, label = below:$s$] {$2\!\!:\!\!2$}
	  edge from parent node [xxx, right] {DRDR}
	};
	\draw [densely dashed] (DR) -- node [xxx, above] {DRD} +(1,-.2);
	\draw [<-, densely dashed] (DRD) -- node [xxx, above] {DRD} +(-1,.2);
    \end{scope}

    \draw [->, densely dashed] (eps) -- node [xxx, above] {D} (D);
    \draw [->, densely dashed] (D) -- node [xxx, above] {DD} (DD);
\end{tikzpicture}
\caption[A range tree]{%
  Part of the range tree for the points from Fig.~\ref{fig: points}.
  The label of node $x$ appears in red on the arrow pointing to $x$.
  Nodes contain $l[x]\!\!:\!\!r[i]$.
  The references $L[x]$ and $R[x]$ appear as children in a binary tree using usual drawing conventions.
  The reference $D[x]$ appears as a dashed arrow (possibly interrupted);
  the placement on the page follows no other logic than economy of layout and readability.
  References $D[x]$ from leaf nodes, such as $D[\mathrm L\mathrm L]$ leading to node LLD, are not shown; this conceals 12 single-node trees.
  The `3rd-dimensional nodes,' whose names contain two Ds, show the values $f[x]$ next to the node.
  To ease comprehension, leaf nodes are decorated with their canonical subset, which is a singleton from $\{p,q,r,s\}$.
  The reader can infer the canonical subset for an internal node as the union of leaves of the subtree; for instance, $P_{\mathrm D\mathrm R} = \{r,s\}$.
  However, note that these point sets are \emph{not} explicitly stored in the data structure.
}
\end{figure}

\begin{lemma}
  Let $n=|P|$.
  Algorithm C computes the $d$-dimensional range tree for $P$ in time linear in $nd\cdot B(n,d)$.
\end{lemma}
\begin{proof}
  We run Algorithm C on input $P$ and $i=1$.

  Disregarding the recursive calls, the running time of algorithm C on input $i$ and $Q$ is dominated by Steps~C2 and C3, i.e., splitting $Q$ into two sets of equal size.
  It is known that this task can be performed in time linear in $|Q|$ \cite{Blum}.
  Thus, the running time for constructing $R_i(Q)$ is linear in $|Q|$ plus the time spent in recursive calls.

  This means that we can bound the running time for constructing $T_1(P)$ by bounding sizes of the sets $P_x$ associated with every node $x$ in the data structure.
  If for a moment $X$ denotes the set of all these nodes then we want to bound \[
    \sum_{x\in X} |P_x| = 
   \sum_{x\in X} |\{\, p\in P\colon p\in P_x\,\}|= 
    \sum_{p\in P} |\{\, x\in X\colon p\in P_x\,\}|
     \,.\]
  Thus, we need to determine, for given $p\in P$, the number of subsets $P_x$ in which $p$ appears.
  By construction, there are fewer than $d$ occurrences of D in $x$.
  Moreover, if $x$ contains more than $h$ occurrences of either L or R then $P_x$ is empty.
  Thus, $x$ has at most $h  + d$ letters.
  For two different strings $x$ and $x'$ that agree on the positions of D, the sets $P_x$ and $P_{x'}$ are disjoint, so $p$ appears in at most one of them.
  We conclude that the number of sets $P_x$ such that $p\in P_x$ is bounded by the number of ways to arrange fewer than $d$ many Ds and at most $h$ non-Ds.
  Using the identity $\binom{a+0}{0} + \cdots +\binom{a+b}{b} = \binom{a+b+1}{b}$ repeatedly, we compute 
    \begin{multline*}
    \sum_{i=0}^{d-1}\sum_{j=0}^h \binom{i+j}{j}=
    \sum_{i=0}^{d-1}             \binom{i+h+1}{h}=
    \sum_{i=0}^{d-1}             \binom{i+h+1}{i+1}=\\
      (-1) + \sum_{i=0}^{d}             \binom{i+h}{i} =
    \binom{h+d+1}{d} -1 = 
    \frac{h+d+1}{h+1}\binom{h+d}{d}-1
    \leq d\binom{d+h}{d}\,.\end{multline*}
  The bound follows from aggregating this contribution over all $p\in P$.
\end{proof}

\paragraph*{Search.}

In this section, we fix two sequences of integers $l_1,\ldots, l_d$ and $r_1,\ldots, r_d$ describing the \emph{query box} $B$ given by
\[ B = [l_1,r_1]\times\cdots\times [l_d,r_d]\,.\]

\medskip\noindent$\blacktriangleright$
{\sffamily\bfseries Algorithm Q}
({\sffamily Query}).
{\it 
Given integer $i\in\{1,\ldots,d\}$, a query box $B$ as above and a  range tree $R_i(Q)$ with root $x$ for a set of points $Q$ such that every point $q\in Q$ satisfies $l_j\leq q_j\leq r_j$ for $j\in \{1,\ldots, i-1\}$.
This algorithm returns $\bigoplus \{\, f(q)\colon q\in Q\cap B\,\}$.
}
\vspace*{-2ex}

\begin{description}\pushQED{\qed}
  \item[Q1] [Empty?] If the data structure is empty, or $l_i> r[x]$, or $l[x]>r_i$, then return the identity in the underlying monoid $M$.
  \item[Q2] [Done?] If $i=d$ and $l_d\leq \min_d[x]$ and $\max_d[x]\leq r_d$ then return $f[x]$.
  \item[Q3] [Next dimension?] 
    If $i<d$ and $l_i\leq l[x]$ and $r[x]\leq r_i$ then query the range tree at $D[x]$ for dimension $i+1$.
    Return the resulting value.
  \item[Q4] [Split.]
    Query the range tree $L[x]$ for dimension $i$; the result is a value $f_L$.
    Query the range tree $R[x]$ for dimension $i$; the result is a value $f_R$.
    Return $f_L\oplus f_R$.\qed
\end{description}

To prove correctness, we show that this algorithm is correct for each point set $Q=P_x$.

\begin{lemma}
  Let $i=D(x)+1$, where $D(x)$ is the number of {\rm D}s in $x$.
  Assume that $P_x$ is such that $l_j\leq p_i\leq r_j$ for all $j\in \{1,\ldots, i-1\}$ for each $p\in P_x$. 
  Then the query algorithm on input $x$ and $i$ returns $f(B\cap P_x)$.
\end{lemma}

\begin{proof}
  Backwards induction in $|x|$.
 
  If $|x| = h+d$ then $P_x$ is the empty set, in which case the algorithm correctly returns the identity in $M$.

    If the algorithm executes Step~Q2 then $B$ is satisfied for all $q \in P_x$, in which case the algorithm correctly returns $f[x] = f(P_x)$.

  If the algorithm executes Step~Q3 then $B$ satisfies the condition in the lemma for $i+1$, and the number of Ds in $P_{xD}$ is $i+1$, and $D[x]$ store the $(i+1)$th range tree for $P_{xD}$. 
  Thus, by induction the algorithm returns $f(P_{x\mathrm D}\cap B)$, which equals $f(P_x\cap B)$ because $P_{xD}= P_x$.

  Otherwise, by induction, $f_L= f(P_{x\mathrm L}\cap B)$ and $f_R=f(P_{x\mathrm R}\cap B)$.
  Since $P_{x\mathrm L} \cup P_{x\mathrm R} = P_x$, we have $f(P_x\cap B) = f((P_{x\mathrm L} \cap B) \cup (P_{x\mathrm R} \cap P)) = f_L\oplus  f_R$. 
\end{proof}

\begin{lemma}
  If $x$ is the root of the range tree for $P$ then on input $i=1$, $x$, and $B$, the query algorithm returns $f(P\cap B)$ in time linear in $2^dB(n,d)$.
\end{lemma}
\begin{proof}
  Correctness follows from the previous lemma. 

  For the running time, we first observe that the query algorithm does constant work in each visited node. Thus it suffices to bound the number of visited nodes as 
  \begin{equation}\label{eq: Q bound}
    2^d\binom{h+d}{d} \qquad (d\geq 1,h\geq 0)\,.
  \end{equation}

  We will show by induction in $d$ that \eqref{eq: Q bound} holds for every call to a $d$-dimensional range tree for a point set $P_x$, where $h=\lceil \log |P_x|\rceil$.
  The two easy cases are Q1 and Q2, which incur no additional nodes to be visited, so the number of visited nodes is $1$, which is bounded by \eqref{eq: Q bound}.
  Step Q3 leads to a recursive call for a $(d-1)$-dimensional range tree over the same point set $P_{xD}=P_x$, and we verify \[ 1+2^{d-1}\binom{h+d-1}{d-1} \leq 2^d\binom{h+d}{d}.\]
  The interesting case is Step~Q4.
  We need to follow two paths from $x$ to the leaves of the binary tree of $x$. 
  Consider the leaves $l$ and $r$ in the subtree rooted at $x$ associated with the points $\min_i(P_x)$ and $\max_i(P_x)$ as defined in Sec.~\ref{sec: Range Trees}.
  We describe the situation of the path $Y$ from $l$ to $x$; the other case is symmetrical.
  At each internal node $y\in Y$, the algorithm chooses Step Q4 (because $l_i\geq l[y]$).
  There are two cases for what happens at $y\mathrm L$ and $y\mathrm R$.
  If $l_i\leq\med_i(P_y)$ then $P_{y\mathrm R}$ satisfies $l_i\leq \min_i(P_{y\mathrm R}) \leq r_i$, so the call to $y\mathrm R$ will choose Step~Q3.
  By induction, this incurs $2^{d-1}\binom{d-1+i}{d-1}$ visits, where $i$ is the height of $y$.
  In the other case, the call to $y\mathrm L$ will choose Step~Q1, which incurs no extra visits.
  Thus, the number of nodes visited on the left path is at most
  \[ h + \sum_{i=0}^{h-1} 2^{d-1} \binom{d-1+i}{d-1}\,,\]
  and the total number of nodes visited is at most twice that:
  \[ 2h + 2^d\sum_{i=0}^{h-1} \binom{d-1+i}{d-1} \leq 2^d\sum_{i=0}^{h} \binom{d-1+i}{d-1} = 2^d \binom{d+h}{d}\,.\]
\end{proof}

\subsection{Discussion}
\label{sec: range tree discussion}

The textbook analysis of range trees, and similar $d$-dimensional spatial algorithms and data structures sets up a recurrence relation like
\[ r(n,d) = 2r(n/2,d) + r(n,d-1)\,,\]
for the construction and
\[ r(n,d) = \max\{\,r(n/2,d) , r(n,d-1)\,\}\,,\]
for the query time.
One then observes that $n\log^d n$ and $\log^d n$ are the solutions to these recurrences.
This analysis goes back to Bentley's original paper \cite{Bentley80}.

Along the lines of the previous section, one can show that the functions $n\cdot B(n,d)$ and $B(n,d)$ solve these recurrences as well.
A detailed derivation can be found in \cite{Monier79}, which also contains combinatorial arguments of how to interpret the binomial coefficients in the context of spatial data structures.
A later paper of Chan~\cite{Chan08} also takes the recurrences as a starting point, and observes asymptotically improved solution for the related question of dominance queries.

\section{Graph Distances}
\label{sec: Graph Distances}

We present the algorithm for computing the diameter.
The construction closely follows Cabello and Knauer~\cite{CK}, but uses the range tree bounds from Section~\ref{sec: Orthogonal Range Queries}.
The analysis is extended to superconstant dimension as in Abboud \emph{et al.}~\cite{AVW}.
Using the approximate treewidth construction of Bodlaender \emph{et al.}~\cite{BDDFLP}, we can pay more attention to the parameters of the recursive decomposition into small-size separators.

\subsection{Preliminaries}

We consider an undirected graph $G$ with $n$ vertices and $m$ edges with nonnegative integer weights.
The set of vertices is $V(G)$.
For a vertex subset $U$ we write $G[U]$ for the induced subgraph.

A path from $u$ to $v$ is called a $u,v$-path and denoted $uPv$.
For $w\in V(uPv)$ we use the notation $wPv$ for the subpath starting in $w$.
The \emph{length} of a path, denoted $l(uPv)$, is the sum of its edge lengths.

The \emph{distance} from vertex $u$ to vertex $v$, denoted $d(u,v)$, is the minimum  length of shortest $u,v$-path.
The \emph{Wiener index} of $G$, denoted $\wien(G)$ is $\sum_{u,v\in V(G)} d(u,v)$.
The \emph{eccentricity} of a vertex $u$, denoted $e(u)$ is given by $e(u) = \max \{\, d(u,v)\colon v\in V(G)\,\}$.
The \emph{diameter} of $G$, denoted $\diam(G)$ is $\max \{\, e(u)\colon u\in V(G)\,\}$.
The \emph{radius} of $G$, denoted $\rad(G)$ is $\min \{\, e(u)\colon u\in V(G)\,\}$.

\subsection{Separation}

A \emph{skew $k$-separator tree} $T$ of $G$ is a binary tree such that each node $t$ of $T$ is associated with a vertex set $Z_t\subseteq V(G)$ such that
\begin{itemize}
  \item $|Z_t| \leq k$,
  \item If $L_t$ $R_t$ denote the vertices of $G$ associated with the left and right subtrees of $t$, respectively, then $Z_t$ separates $L_t$ and $R_t$ and 
    \begin{equation}\label{eq: |S| bound} 
    \frac{n}{k+1}\leq |L_t\cup Z_t|\leq \frac{nk}{k+1}\,,
    \end{equation}
  \item $T$ remains a skew $k$-separator even if edges between vertices of $Z_t$ are added.
\end{itemize}

It is known that such a tree can be found from a tree decomposition, and an approximate tree decomposition can be found in single-exponential time.
We summarise these results in the following lemma:

\begin{lemma}[{\cite[Lemma 3]{CK}} with {\cite[Theorem 1]{BDDFLP}}]
  \label{lem: sep}
  For a given $n$-vertex input graph $G$, a skew $(5\tw(G)+4)$-separator tree can be computed in time $n\exp O(k)$.
\end{lemma} 

\subsection{Algorithm}

Given graph $G$, let $\mathscr S$ denote the set of shortest paths. 
Let $e(x;W)$ denote the distance from $x$ to any vertex in $W$.
Formally, \[ e(x;W)   = \max \{\, l(xPw) \colon xPw\in \mathscr S, w\in W\,\}\,.\]

The central idea of the algorithm, following \cite{CK}, is the computation for $x\in X$, $z\in Z$ of \emph{$z$-visiting} eccentricities $e(x,z;Y)$ defined as follows.
Enumerate $Z= \{z_1,\ldots, z_k\}$.
Then define, for $x\in X$, $z_i\in Z$ the value $e(x,z_i;Y)$ as the maximum distance from $z_i$ to $y$ over all $y \in Y$ such that some shortest $x,y$-path contains $z_i$ but no shortest $x,y$-path contains any of $\{z_1,\ldots,z_{i-1}\}$. 
Formally,
\begin{align*}
  e(x,z_i;Y)  = & \max l(zPy)\\
  \text{such that } & y\in Y\,,\\
  & xPy\in \mathscr S\,,\\
  & Z\cap  V(xPy) \ni z_i\,,\\
  & \{z_1,\ldots, z_{i-1}\}\cap V(xQy) = \emptyset \text{ for all } zQy\in \mathscr S \,.
\end{align*}
See Figure~\ref{fig: e} for a small example. 

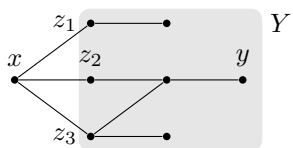
\begin{figure}
\tikzstyle{vertex}=[fill, circle, inner sep = 1pt]
\begin{tikzpicture}[yscale = .75]
  \draw [rounded corners,  fill, gray!20] (.85, -1.25) rectangle (3.25,1.25);
  \node (x) at (0,0)  [vertex, label = above:$x$]   {};
  \node (1) at (1,+1) [vertex, label = left:$z_1$]  {};
  \node (2) at (1,0)  [vertex, label = above:$z_2$] {};
  \node (3) at (1,-1) [vertex, label = left:$z_3$]  {};
  \node (4) at (2,+1) [vertex] {};
  \node (5) at (2,0)  [vertex] {};
  \node (6) at (2,-1) [vertex] {};
  \node (y) at (3,0)  [vertex, label = above:$y$] {};
  \draw (x)--(1);
  \draw (x)--(2);
  \draw (x)--(3);
  \draw (1)--(4);
  \draw (2)--(5);
  \draw (3)--(5);
  \draw (3)--(6);
  \draw (5)--(y);
  \node at (3.5, 1) {$Y$};
\end{tikzpicture}
  \caption{\label{fig: e}
  Example with $Z=\{z_1,z_2,z_3\}$.
  The eccentricity of $x$ to $Y$ is $e(x;Y)=3$. 
  Also, $e(x,z_1;Y)=1$, $e(x,z_2;Y) = 2$.
  Note $e(x,z_3; Y)=1$ despite the $z_3,y$-path.}
\end{figure}

This definition ensures that in situations where $x$ and $y$ are connected by two shortest paths of the form $xPz_jPy$ and $xPz_iPy$ with $j\neq i$, then exactly one of them contributes to $e(x,z_j;Y)$ and $e(x,z_i;Y)$.
This is important for avoiding over-counting in Section~\ref{sec: Wiener Index}.

\begin{lemma}\label{lem: e(x;Y)}
  For $x\in X$,
  $e(x;Y)= \max\{\, d(x,z)+ e(x,z;Y) \colon z\in Z\,\}$.
\end{lemma}

The proof is in Appendix~\ref{sec: proof of e(x;Y)}.
The connection to orthogonal range queries is the following.
Enumerate $Z=\{z_1,\ldots, z_k\}$.
A shortest path $xPz_iPy$ attaining the distance $e(x;Y)$ maximises \( d(z_i,y) \) over all $y\in Y$, where $z_i\in Z$ is such that for all $z_j\in Z$, 
\begin{align*}
  d(x,z_i) + d(z_i,y) & < d(x,z_j) + d(z_j,y)\,, \qquad (j<i)\,,\\
  d(x,z_i) + d(z_i,y) &\leq d(x,z_j) + d(z_j,y)\,,\qquad (j\geq i)\,, 
\end{align*}
equivalently,
\begin{align*}
  d(x,z_i) - d(x,z_j)  &<  d(z_j,y)-d(z_i,y)\,,\qquad (j<i)\,,\\
  d(x,z_i) - d(x,z_j)  &\leq  d(z_j,y)-d(z_i,y)\,,\qquad (j\geq i)\,.
\end{align*}

\medskip

We are ready for the algorithm, which closely follows \cite{CK}:

\bigskip\noindent$\blacktriangleright$
{\sffamily\bfseries Algorithm E}
({\sffamily Eccentricities}).
{\it 
Given a graph $G$ and a skew $k$-separator tree with root $t$, this algorithm computes the eccentricity $e(v)$ of every vertex $v\in V(G)$.
We write $Z=Z_t$, $X= L_t\cup Z_t$, and $Y=R_t \cup Z_t$.
}

\vspace*{-2ex}
\begin{description}
  \item[E1] [Base case.] 
    If $n/\ln n< 4k(k+1)$  then find all distances using Dijkstra's algorithm and terminate.
  \item[E2] [Distances from separator.] 
    Compute $d(z,v)$ for every $z\in Z,v\in V(G)$ using $k$ applications of Dijkstra's algorithm.
  \item[E3] [Add shortcuts.] 
    For each pair $z,z'\in Z$, add the edge $zz'$ to $G$, weighted by $d(z,z')$.
    Remove duplicate edges, retaining the shortest.
  \item[E4.1] [Start iterating over $\{z_1,\ldots,z_k\}$.]
    Let $i=1$.
  \item[E4.2] [Build range tree for $z_i$.]
    Construct a $k$-dimensional range tree $R$ of points $\{\, p(y)\colon y\in Y\,\}$ given by $p(y) = (p_1,\ldots, p_k)$ where
    \[ p_j = d(z_i,y) - d(z_j,y)\qquad j\in\{1,\ldots,k\} \,\]
    and $f(p(y)) = d(z_i,y)$ using the monoid $(\mathbf Z,\max)$.
  \item[E4.3] [Query range tree.]
    For each $x\in X$, query $R$ with $l_1=\dots = l_k = -\infty$ and
    \[ r_j=
    \begin{cases}
      d(x,z_i)-d(x,z_j)-1\,, & (j< i)\,;\\
      d(x,z_i)-d(x,z_j)\,, & (j\geq i)\,.
    \end{cases}
    \]
    The result is $e(x,z_i;Y)$.
  \item[E4.4] [Next $z_i$.]
    If $i<k$ then increase $i$ and go to E4.1.
  \item[E5] [Recurse on $G[X]$.]
    Recursively compute the distances in $G[X]$ using the left subtree of $t$ as a skew $k$-separator tree.
    The result are eccentricities $e(x;X)$ for each $x\in X$.
    For each $x\in X$, set $e(x;Y)= \max \{\, d(x,z_i) + e(x,z_i;Y) \colon i\in\{1,\ldots, k\}\,\}$, then set $e(x) = \max \{ e(x;X), e(x;Y)\}$.
    Set $e(z)=\max \{ \, d(z,v) \colon v\in V(G)\,\}$ for $z\in Z$.
  \item[E6] [Flip.]
    Repeat Steps E4--5 with the roles of $X$ and $Y$ exchanged.
\end{description}

\subsection{Running Time}

\begin{lemma}\label{lem: running time E}
  The running time of Algorithm E is $O(n\cdot B(n,k)\cdot 2^k k^2\ln n)$.
\end{lemma}

The proof is in Appendix~\ref{sec: proof E}.
We can now establish Theorem~\ref{thm: main} for diameter and radius.

\begin{proof}[Proof of Thm.~\ref{thm: main}, distances]
  To compute all eccentricities for a given graph we find a $k$-skew separator for $k=5\tw(G)+4$ using Lemma~\ref{lem: sep} in time $n\exp O(\tw (G))$.
  We then run Algorithm E, using Lemma~\ref{lem: running time E} to bound the running time.
  From the eccentricities, the radius and diameter can be computed in linear time using their definition.
\end{proof}

\subsection{Wiener Index}
\label{sec: Wiener Index}

Algorithm E can be modified to compute the Wiener index, as described in \cite[Sec.~4]{CK}, completing the proof of Theorem~\ref{thm: main}.
The main observation is that the sum of distances between all pair $u,v\in V(G)$ can be written as pairwise distances within $X$, within $Y$, and between $X$ and $Y$, carefully subtracting contributions from these sums that were included twice.

The orthogonal range queries for vertex $x\in X$ now need to report the sum of distances to every $y\in Y$, rather than just the value of the  maximum distance $e(x;Y)$.
To this end, we use the monoid of positive integer tuples $(d,r)$ with the operation \[
  (d,r)\oplus (d',r') = (d+d', r+r')
\]
with identity element $(0, 0)$.
The value associated with vertex $y$ in Step~E4.2 is $f(p(y)) = (1, d(z_i,y))$.

\medskip
We also observe the matching lower bound:

\begin{theorem}
  An algorithm for computing the Wiener index in time $n^{2-\epsilon}\exp o (\tw(G))$ time for any $\epsilon>0$ refutes the Orthogonal Vector conjecture.
\end{theorem}
\begin{proof}
  The diameter of $G$ is $2$ if and only if $\wien(G)= 2\binom{n}{2}  -m$.
    Thus, an algorithm for Wiener index is able to distinguish input graphs of diameter 2 and 3.
    This problem was shown hard in \cite{AVW}.
\end{proof}

\bibliography{main}
\appendix
\section{Parameterization by Vertex Cover Number}

We show Theorem~\ref{thm: vertex cover main}.

\newcommand{\vc}{\operatorname{vc}}

A \emph{vertex cover} is a vertex subset $C$ of $V(G)$ such that every edge in $G$ has at least one endpoint in $C$.
The smallest $k$ for which a vertex cover of size $k$ exists is the \emph{vertex cover number} of a graph, denoted $\vc(G)$.
The number of edges in a graph is at most $n\cdot\vc(G)$.

\subsection{Eccentricities and Wiener Index}

A graph with vertex cover number $1$ is a star, and its pairwise distances can be determined from the input size.
It follows from \cite{AVW} that the complexity of computing the diameter must depend exponentially on $\vc(G)$, in the same way as for $\tw(G)$.
We observe here that algorithms that match this lower bound are quite immediate.
The idea is that each $v\notin C$ has its entire neighbourhood $N(v)$, defined as 
\[ N(v)= \{\, u\in V(G)\colon uv\in E(G)\,\}\,,\]
contained in $C$.
Thus, all paths from $v$ have their second vertex in $C$.
In particular, two vertices $v$ and $w$ with $N(v)= N(w)$ have the same distances to the rest of the graph.
Since $N(v)\subseteq C$ is suffices to consider all $2^k$ many subsets of $C$.
The details are given in Algorithm~V.

\bigskip\noindent$\blacktriangleright$
{\sffamily\bfseries Algorithm V}
({\sffamily Eccentricities Parameterized by Vertex Cover}).
{\it 
Given a connected, unweighted, undirected graph $G$ and a vertex cover $C$, this algorithm computes the eccentricity of each vertex and the Wiener index.
}

\vspace*{-2ex}
\begin{description}
  \item [V1] [Initialise.] 
    Set $\wien(G)=0$.
    Insert each $v\notin C$ into a dictionary $D$ indexed by $N(v)$.
  \item [V2] [Distances from $C$.]
    For each $v\in C$, perform a breadth-first search from $v$ in $G$, computing $d(v,u)$ for all $u\in V(G)$.
    Let $e(v) = \max_u d(v,u)$ and increase $\wien(G)$ by $\frac{1}{2}\sum_u d(v,u)$.
  \item [V3] [Distances from $V(G)-C.$]
    Choose any $v\in D$.
    Perform a breadth-first search from $v$ in $G$, computing $d(v,u)$ for all $u\in V(G)$.
    For each $w\in D$ with $N(w)= N(v)$ (including $v$ itself), let $e(w)=\max_u d(v,u)$, increase $\wien(G)$ by $\frac{1}{2} \sum_u d(v,u)$, and remove $w$ from $D$.
    Repeat step V3 until $D$ is empty.
\end{description}

\begin{theorem}
  \label{thm: vc1}
  The eccentricities and Wiener index of an unweighted, undirected, connected graph with $m$ edges and vertex cover number $k$ can be computed in time $O(m2^k)$.
  Any algorithm with running time $m\exp o(k)$ would refute the Strong Exponential Time Hypothesis.
\end{theorem}

\begin{proof}
  It is well known that a minimum vertex cover can be computed in the given time bound \cite{DF}. 

  For the running time of Algorithm~V, we first observed that for each $v\notin C$ the neighbourhood $N(v)$ is entirely contained in $C$.
  Thus, there are only $2^k$ different neighbourhoods used as an index to $D$ and we can bound the number of BFS computations in Step~V3 by $2^k$.
  (Step V2 incurs another $k$ such computations.)
  Assuming constant-time dictionary operations, the running time of the algorithm is therefore bounded by $m\exp O(k)$.

  To see correctness, we need to argue that the distances computed for $w\in D$ in Step 4 are correct.
  First, to argue $d(v,z)=d(w,z)$ for all $z\notin \{v,w\}$ consider shortest path $vPz$ and $wQz$.
  Let $v'$ and $w'$ denote the second vertices on these paths (possibly $v'=z$ or $w'=z$).
  Then $v'Pz$ and $w'Qz$ are shortest paths of length $l(vPz)-1$ and $l(wQz)-1$, respectively.
  Since $N(v)=N(w)$, the path $wv'Pz$ exists and is a shortest $w,z$-path as well, of lengths $l(vPz)$.
  We conclude that $l(vPz)=l(wQz)$.

  It is \emph{not} true that $d(v,z)=d(w,z)$ for $z\in\{v,w\}$.
  Instead, we have $d(v,w)=d(w,v)$ (both equal 2)
  and $d(v,v) = d(w,w)$ (both equal 0).
  Thus, the contributions from $v$ and $w$ to $W$ are the same, and the sets $d(v,\cdot)$ and $d(w,\cdot)$ have the same maxima.

  For the hardness result, we merely need to observe that reduction in \cite{AVW} has vertex cover number $k+2$.
\end{proof}

\subsection{Faster Eccentricities}

Vertex cover number is an extremely well studied parameter, so the analysis need not stop here.
The best current algorithm for \emph{finding} a vertex cover runs in time $O(nk + 1.274^k)$ \cite{Chen06}, so the bound in Theorem~\ref{thm: vc1} is dominated by the distance computation.
Thus it may make sense to look for distance computation algorithms with running times of the form $nk + g(k)$ rather than $m\cdot g(k)$.

\bigskip
We can find such an algorithm for eccentricities, but not for Wiener index.
First, we observe that if $C$ is a vertex cover then no path can contain consecutive vertices from $V(G)-C$.
Thus, we can modify the graph by inserting length-2 shortcuts between nonadjacent vertices in $C$ that share a neighbour without changing the pairwise distances in the graph.
We can now run Dijkstra restricted to the subgraph $G[C\cup \{v\}]$, noting that the second layer of the Dijkstra tree consists of $N(v)$, which is contained in $C$.
Thus the number of such computations that are different is bounded by $2^k$, the number of neighbourhoods.
The eccentricity $e(v)$ can be derived from this Dijkstra tree as follows.
Let $E(v)$ denote the \emph{eccentric} vertices from $v$ in $C$, i.e., the vertices at maximum distance from $v$ in $C$.
Note that $E(v)$ contains exactly the vertices at the deepest layer of the Dijkstra tree from $v$ in $G[C\cup \{v\}]$.
The only vertices $u$ in $G$ that can be farther away from $v$ than $E(v)$ must have their entire neighbourhood $N(u)$ contained in $E(v)$.
See Figure~\ref{fig: E(v)}.

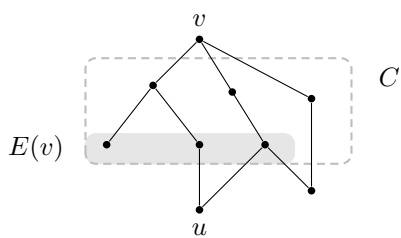
\begin{figure}
  \tikzstyle{vertex}=[fill, circle, inner sep = 1pt]
  \begin{tikzpicture}[node distance  = 0.75]
    \node (v) [vertex, label=above:$v$] at (0,.5) {};
    \draw [rounded corners, fill, gray!20] (-1.5,-1.15) rectangle (1.25,-.75);
    \node at (-2.15, -.95) {$E(v)$};
    \draw [rounded corners, draw, gray!50, thick, densely dashed] 
    	(-1.5,-1.15) rectangle (2,0.25);
    \node at (2.5, 0) {$C$};
    \node (a) [vertex, below left  = of v] {};
    \node (1) [vertex, below left  = of a, yshift = -5pt]  {};
    \node (2) [vertex, below right = of a, yshift = -5pt]  {};
    \node (3) [vertex, right = of 2] {};
    \node (u) [vertex, below = of 2, label=below:$u$] {};
    \node (6) [vertex, above right = of 3] {};
    \node (5) [vertex, below right = of 3] {};
    \draw (u)--(3);
    \draw (u)--(2);
    \draw (5)--(3);
    \draw (5)--(6);
    \draw (v)--(6);
    \draw (v)--(a);
    \draw (a)--(1);
    \draw (a)--(2);
    \draw (v)-- node [vertex] {}  (3);
  \end{tikzpicture}
  \caption{\label{fig: E(v)} The distance from $v$ to $u$ is $3$ and $N(u)\subseteq E(v)$.}
\end{figure}

The only confusion arises if the only such vertex is $v$ itself.
To handle these details we need to determine, for each cover subset $S\subseteq C$, if the number of $u$ with $N(u)\subseteq S$ is $0$, $1$, or more.
This can be solved by fast zeta transform in time $2^kk$, or more directly as follows.
For each $S\subseteq C$, let 
\[ h(S) = 
\begin{cases}
  \{w\}\,, & 
  \text{if $N(w)\subseteq S$ for exactly one $w\notin C$}\,;\\
  \emptyset\,, & 
  \text{if there is no $w\notin C$ with $N(w)\subseteq S$}\,;\\
  C\,, & 
  \text{otherwise}\,.
\end{cases}
\]
(The third value is an arbitrary placeholder.)
Then $h(S)$ can be computed for all $S\subseteq C$ in a bottom-up fashion.

The details are given in the following algorithm.

\bigskip\noindent$\blacktriangleright$
{\sffamily\bfseries Algorithm F}
({\sffamily Faster Eccentricities Parameterized by Vertex Cover}).
{\it 
Given a connected, unweighted, undirected graph $G$ and a vertex cover $C$, this algorithm computes the eccentricity of each vertex.
}

\begin{description}
  \item [F1] [Initialise.]
    Insert each $v\in V(G)$ into a dictionary $D$ indexed by $N(v)$.
    Set $h(S)=\emptyset$ for all $S\subseteq C$.
  \item [F2] [Compute $h$.]
    For each $u\notin C$, set $h(N(u))=\{u\}$ if $h(N(u))=\emptyset$, otherwise $h(N(u)) = C$.
    For each nonempty subset $S \subseteq C$ in increasing order of size, compute $W=\bigcup_{w\in S} h(S-\{w\})$.
    If $|W|>1$ then set $h(S)=C$.
    Else set $h(S)=W$.
  \item [F3] [Shortcuts.]
    For each pair of covering vertices $u,v\in C$, if $uv\notin E(G)$ but $u$ and $v$ share a neighbour  outside $C$, add the edge $uv$ to $E(G)$ with length $2$.
  \item [F4] [Eccentricities from $C$.]
    For each $v\in C$, compute shortest distances in  $G[C]$ from $v$.
    Set $d = \max_{u\in C} d(v,u)$ and let $E(v)$ denote the vertices in $C$ at distance $d$. 
    Let \[
      e(v)=
      \begin{cases}
	d + 1 \,, & \text{if } h(E(v))-\{v\}\neq\emptyset\,, \text{ [equivalently, $E(v)\supseteq N(w)$ for some $w\neq v$]}\,;\\	
	d\,,      &\text{otherwise}\,.
      \end{cases}
    \]
  \item [F5] [Eccentricities from $V(G)-C$.]
    For each $v\in D$, compute shortest distances in $G[C\cup \{v\}]$ from $v$.
    Set $d = \max_{u\in C} d(v,u)$ and let $E(v)$ denote the vertices in $C$ at distance $d$. 
    For each $u\in D$ with $N(u) = N(v)$ (including $v$ itself) let \[
      e(u)=
      \begin{cases}
	d + 1 \,, & \text{if } h(E(u))-\{u\}\neq\emptyset\,, \text{ [equivalently, $E(u)\supseteq N(w)$ for some $w\neq u$]}\,;\\	
	d\,,      &\text{otherwise}\,.
      \end{cases}
    \]
    and remove $u$ from $D$. 
\end{description}

\begin{theorem}
  The eccentricities an unweighted, undirected, connected graph with $m$ edges and vertex cover number $k$ can be computed in time $O(nk+ 2^kk^2)$.
\end{theorem}
\begin{proof}
  Step F1 needs to visit every of the $nk$ edges.
  There are $2^k$ subsets of $C$, bounding the running time of Step~F2 to $O(2^kk)$.
  Step F3 can be performed in time $O(2^k k^2)$ (instead of the obvious $O(nk^2)$) by iterating  over $w \in D$ and all pairs $u,v \in N(w)$.
  The shortest path computations in Steps~F4 and F5 take time $O(k^2)$ each using Dijkstra's algorithm, for a total of $O(2^kk^2)$. 
  The dictionary contains at most $n$ values, so the total time of Step~F4 and F5 is $O(n+2^kk^2)$.

  \medskip
  To see correctness, assume without loss of generality that we already performed the shortcut operation in Step~F3.

  We argue for correctness of Step~F5, Step~F4 is similar.
  Consider a shortest $u,v$-path $uPv$ to an eccentric vertex $v$ of $u$.
  If $v\in C$ then $v$ belongs to $E(v)$.
  Moreover, there can be no vertex $w$ with $N(w)\subseteq E(v)$, because otherwise $uPvw$ is a shortest path and therefore $v$ is not eccentric.
  Thus, Step~F5 correctly sets $e(u)$ to $d(u,v)$.

  Otherwise, assume all such paths have $v\notin C$.
  There are two cases.
  If $uPv$ is just the edge $uv$ then every vertex in $G$ has distance at most $1$ to $u$.
  If $G$ is a star then $C= N(v) =  \{u\}$ and $d=0$.
  Moreover, $h(E(u)) \ni v$, so Step~F5 correctly computes $e(u) =d+1=1$.
  If $G$ contains a triangle then $|C|>1$ and the vertices in $E(u)$ are at distance $1$.
  Moreover, there cannot exist $w\neq u$ with $N(w)\subseteq E(u)$ because then there would be a $u,w$-path of length 2.
  Thus, Step~F5 correctly computes $e(u)= d= 1$.

  The remaining case is when $uPv$ contains at least 3 vertices.
  Let $w$ denote the penultimate vertex, so the path is of the form $u\cdots wv$.
  Since $v\notin C$ we have $w\in C$.
  Moreover, we have $w\in E(u)$.
  (Otherwise there would be an eccentric path to another vertex $w'\in C$.)
  Let $d=d(u,w)$.
  Every neighbour $x$ of $v$ must belong to $C$, and by the shortcutting Step~F3, we can assume it also belongs to $N(w)\cup \{w\}$.
  The distance $d(v,x)$ is at most $d+1$ (because it is a neighbour of $w$, or $w$ itself), but cannot be $d+1$ (because then there would be an eccentric $u,x$-path for $x\in C$.)
  Thus, we have $d(u,x)=d$ and therefore $x\in E(u)$.
  We have established that $N(v)\subseteq E(u)$, so we can again conclude that Step~F5 correctly computes $e(u) = d+1$.
\end{proof}

\section{Proof of Lemma \ref{lem: e(x;Y)}}
\label{sec: proof of e(x;Y)}

\begin{proof}
  Let $xPy$ be a shortest path of length $e(x;Y)$.
  Since $Z$ separates $X$ from $Y$, any $x,y$-path must contain a vertex from $Z$.
  In particular, this is true of $xPy$, so we can choose $z_i\in Z\cap V(xPy)$ for some $i\in\{1,\ldots, k\}$.
  Assume $xPy$ was chosen so as to minimize $i$.
  Since $xPz_i$ is a shortest path, we have $l(xPz_i) = d(x,z_i)$.
  Moreover, $z_iPy$ is a shortest path, and there are no shortest $x,y$-paths through $\{z_1,\ldots, z_{i-1}\}$, so $l(z_iPy)\leq e(x,z_i;Y)$. Thus $e(x;Y)\leq d(x,z_i)+e(x,z_i;Y)$ for some $z_i\in Z$.

  For the opposite inequality, choose any $z\in Z$ and shortest paths $xPz$ and $xQy$ with $z\in V(xQy)$ such that $l(xPz)= d(x,z)$ and $l(zQy)= e(x,z;Y)$.
  Since $xQz$ is a shortest path, we see that \[
    d(xPz)+e(x,z;Y)  = l(xPzQy) =l(xQzQy)=l(xQy)\,,\]
  which is the length of a shortest $x,y$-path, and therefore at most $e(x;Y)$.
\end{proof}

\section{Proof of Lemma~\ref{lem: running time E}}
\label{sec: proof E}

\begin{proof}
  Assume $n\geq 8$.
  Let $T(n,d)$ denote the running time of Algorithm E.

  Step E1 consists of $n$ executions of Dijkstra's algorithm on a graph with $n$ vertices and treewidth $O(k)$, and $n$ bounded by $O(k^2\log k)$.
  This takes time $O(k^5\log^3 k)$.
  The range query operations in Steps E4.2--3 can be performed in time $O(n2^k\cdot B(n,k))$ according to to Lemma~\ref{lem: main}.
  They are executed $2k$ times, twice for each $z_i\in Z$.
  Thus, adding the recursive calls in step E5 for both $X$ and $Y$ using $|Y| \le n - |X| + k$, we arrive at the recurrence
  \[ T(n,k) =\begin{cases}
    O(k^5\log^3 k) \,, &\text{if } n/\ln n< 4k(k+1)\,;\\
     n\cdot S(n,k) + T(|X|, k) + T(n-|X|+k,k)\,, & \text{otherwise.}
  \end{cases}
  \]
  for some non-decreasing function $S$ satisfying 
  \(S(n,k) = O\bigl(2^kk\cdot B(n,k)\bigr)\,.\)

  We will show
  \[ T(n,k) \le 4 (k+1)\cdot S(n,k)\cdot  n\ln n\,. \]
  Write $s=|X|$ and $r=n-s+k$, and consider
  \begin{equation}\label{eq: derivation}
    \frac{T(s,k) + T(r,k)}{4(k+1)} =
   S(s,k)\cdot s\ln s + S(r,k)\cdot r\ln r \leq
  S(n,k)\cdot (s\ln s +r\ln r)\,.
  \end{equation}
  From the bounds \eqref{eq: |S| bound} on $s$ we have  $s\leq nk/(k+1)$ and $r \leq n-n/(k+1) +k =  (nk/(k+1)) +k$, so we can bound both $r$ and $s$ by $t$ given as
  \[t=\frac{nk}{k+1} +k\,.\]
  Thus we can bound \eqref{eq: derivation} by
  \begin{equation}\label{eq: logs}
    \frac{T(s,k)+T(r,k)}{4(k+1)\cdot S(n,k)} \leq
    s\ln t +(n-s+k)\ln  t = n\ln t + k\ln t
  \leq n\ln t + k\ln n\,.
  \end{equation}
  Step E1 ensures
  \( k(k+1)\leq n/(4\ln n) \leq \textstyle\frac{1}{2}n\,,\)
  so we get
  \[ t \leq \frac{nk}{k+1} + k \leq \frac{nk}{k+1} +\frac{n/2}{k+1}\,.\]
  Using the bound \(\ln y\leq 1/(y-1)\) for $y \in (0,1)$, we see \( \ln ((k+\frac{1}{2})/(k+1)) \leq - 1/(2k+2)\), so
  \begin{equation}\label{eq: bound on ln t}
    \ln t \leq
    \ln \biggl(\frac{n(k+\frac{1}{2})}{k+1} \biggr) \leq
    \ln n - \frac{1}{2k+2} 
    \,,
  \end{equation}
  Using this in \eqref{eq: logs}, we have
  \[
    \frac{T(s,k)+T(r,k)}{4(k+1)\cdot S(n,k)} \leq
    n\ln n - \frac{n}{2k+2}+k\ln n \,.\]
  The last term satisfies
  \( k\ln n \leq n/4(k+1) \)
  because of the guarantee in Step~E1.
  Thus, 
  \[ 
  \frac{T(s,k)+T(r,k)}{4(k+1)\cdot S(n,k)} + \frac{n}{4(k+1)} \leq n\ln n\,,\]
  so that $T$ satisfies the recurrence.
\end{proof}
\end{document}